\documentclass[journal,twoside,web]{ieeecolor}
\usepackage{lcsys}
\usepackage{cite}
\usepackage{amsmath,amssymb,amsfonts,amstext}
\usepackage{graphicx}
\usepackage{subcaption}
\usepackage{textcomp}
  \def\BibTeX{{\rm B\kern-.05em{\sc i\kern-.025em b}\kern-.08em
     T\kern-.1667em\lower.7ex\hbox{E}\kern-.125emX}}
\usepackage{bbold}
\usepackage{bbm}
\usepackage{mathrsfs}
\usepackage{verbatim}
\usepackage{amsmath}
\usepackage{amsfonts}
\usepackage{bm}
\usepackage{caption}
\usepackage{float}

\usepackage{amsthm}
\usepackage{mathtools}
\usepackage{cite}
\usepackage{hyperref}

\allowdisplaybreaks

\usepackage{algorithm}
\usepackage{algpseudocode}

\theoremstyle{definition}
\newtheorem{assumption}{Assumption}
\newtheorem{definition}{Definition}
\newtheorem{remark}{Remark}
\newtheorem{problem}{Problem}

\theoremstyle{plain}

\newtheorem{theorem}{Theorem}
\newtheorem{lemma}{Lemma}

\makeatletter
\let\NAT@parse\undefined
\makeatother
\usepackage{hyperref}
\hypersetup{
  colorlinks=true,
    linkcolor= blue,
    allcolors=blue,
    citecolor = blue,
    filecolor=black,      
    urlcolor=blue,
    }
\title{A Framework for Effective AI Recommendations in Cyber-Physical-Human Systems}

\author{Aditya Dave, \textit{Member, IEEE}, Heeseung Bang, \textit{Student Member, IEEE},\\ Andreas A. Malikopoulos, \textit{Senior Member, IEEE}
    \thanks{This research was supported by NSF under Grants CNS-2149520 and CMMI-2219761.}
    \thanks{The authors are with The School of Civil and Environmental Engineering, Cornell University, Ithaca, NY 14850, USA. {\tt\small email: \{a.dave,hb489,amaliko\}@cornell.edu}}
}

\date{March 2024}
\begin{document}

\maketitle
\thispagestyle{empty}
\begin{abstract}
Many cyber-physical-human systems (CPHS) involve a human decision-maker who may receive recommendations from an artificial intelligence (AI) platform while holding the ultimate responsibility of making decisions. In such CPHS applications, the human decision-maker may depart from an optimal recommended decision and instead implement a different one for various reasons. In this letter, we develop a rigorous framework to overcome this challenge.
In our framework, we consider that humans may deviate from AI recommendations as they perceive and interpret the system's state in a different way than the AI platform. We establish the structural properties of optimal recommendation strategies and develop an approximate human model (AHM) used by the AI. We provide theoretical bounds on the optimality gap that arises from an AHM and illustrate the efficacy of our results in a numerical example.
\end{abstract}

\begin{IEEEkeywords}
Cyber-Physical Human Systems, Human-AI Interaction, Human Model, Recommender Systems.
\vspace{-11pt}
\end{IEEEkeywords}

\section{Introduction}
In several cyber-physical-human systems (CPHS), e.g., aircraft co-pilot \cite{uzun2023enhancing}, autonomous driving \cite{Nishanth2023AISmerging},  social media \cite{Dave2020SocialMedia},  a human decision-maker may receive recommendations from an artificial intelligence (AI) platform while holding the ultimate responsibility of making decisions. For example, consider a traffic environment \cite{Bang2023mem} where a human driver receives a recommendation for following a particular route to avoid congestion by a central traffic management system running by an AI platform. In such CPHS applications, the human decision-maker may depart from an optimal recommended decision and instead implement a different one for various reasons \cite{green2019principles}. For example, the human decision-maker may (1) perceive and interpret the system's observations in a different way than the AI platform; (2) have different objectives or restrictions than those designated for the AI; or (3) have more confidence in their inherent decision-making ability or be averse to implementing the suggestions of an algorithm.
Thus, CPHS pose additional challenges \cite{Malikopoulos2022a} to their control because of the influence of humans within the decision-making loop \cite{samad2023human}.

To better understand this phenomenon, there has been recent interest in learning \cite{carroll2019utility} and empirically developing models for human behavior \cite{annaswamy2023human} during collaborations with AI platforms.
It has been established that humans are likely to adhere to recommendations that are easy to interpret and reaffirm their preconceived opinions \cite{dietvorst2018overcoming}.
Furthermore, there is evidence that humans may mistrust AI suggestions, disregard recommendations that can cause discomfort \cite{sun2022predicting}, or misinterpret recommendations \cite{balakrishnan2022improving}, worsening the overall system performance \cite{sabate2003adherence}.
In response to these findings, many research efforts have focused on developing approaches to increase human trust towards AI platforms \cite{glikson2020human} and increase human adoption of AI recommendations \cite{dietvorst2015algorithm}.
However, there remains a need to design principled approaches that an AI platform can use to account for human behavior when generating recommendations.

The adherence-aware Markov decision process is one approach to formalize these human-AI interactions by limiting human behavior to two choices: they may either accept or reject AI suggestions at each instance of time, as dictated by their adherence probability \cite{grand2022best}. In this context, optimal recommendations can be derived for humans with unknown adherence probabilities in unknown environments using Q-learning \cite{faros2023adherence}. Furthermore, this framework has motivated reinforcement learning approaches that explicitly consider whether an AI platform should abstain from recommending decisions \cite{chen2023learning}.
While promising, each of these results relies upon the specific model of human behavior and assumes a system with a perfectly observed Markovian state.
These assumptions will not hold for most CPHS applications. Consequently, there is a need for more general approaches to this problem. 

In this letter, we present a general framework for effective AI recommendations to humans in partially observed CPHS.
We impose minimal assumptions on human behavior and develop our theory to support both empirical modeling and learning from human interactions.  
Our contributions in this letter are (1) a framework for AI recommendations in CPHS and a derivation of the structure of optimal recommendation strategies (Theorem \ref{theorem:optimal_recommendation}), and (2) the introduction of an ``approximate human model" (Definition \ref{definition:ahm}) that yields approximately optimal recommendation strategies with guaranteed performance bounds (Theorem \ref{theorem:approximation_garuantees}). 
We also illustrate the efficacy of our framework in a numerical example.



The remainder of the letter proceeds as follows. In Section \ref{section:problem_formulation}, we present our formulation. In Section \ref{section:recommendation_framework}, we analyze the structure of optimal recommendations, propose an approximate human model, and derive approximation bounds. In Section \ref{sec:example}, we present a numerical example, and in Section \ref{sec:conclusion}, we draw concluding remarks.

\vspace{-6pt}

\section{Problem Formulation} \label{section:problem_formulation}

We consider an AI platform that recommends decisions to a human in a CPHS. The human is responsible for implementing actions that influence the system's evolution. In this context, the human implements a decision by incorporating the platform's recommendations with an instinctive understanding of the situation, as illustrated in Fig. \ref{fig:problem_setup}.
Thus, the AI platform must account for the possibility that a human may re-interpret or disregard the recommended actions.
The CPHS has a finite state space $\mathcal{X}$, and the human selects actions from a finite feasible set $\mathcal{U}$. The system evolves over discrete time steps until a finite horizon $T \in \mathbb{N}$. At each time $t \in \mathcal{T} = \{0,1,\dots,T\}$, the state of the system is denoted by the random variable $X_t \in \mathcal{X}$ and the action implemented by the human is denoted by the random variable $U_t^{\text{h}} \in \mathcal{U}$.
Starting at the initial state $X_0 \in \mathcal{X}$, the evolution of the system at each $t \in \mathcal{T}$ is described by $X_{t+1}=f(X_t,U_t^{\text{h}},W_t)$, 
where $W_t$ is a random variable that corresponds to the external, uncontrollable disturbance and takes values in a finite set $\mathcal{W}$. The disturbances form a sequence of independent random variables
$\{W_t: t \in \mathcal{T}\}$ that are also independent of the initial state $X_0$.
At each $t \in \mathcal{T}$, the system output is denoted by the random variable $Y_t$ taking values in a finite set $\mathcal{Y}$. The output is described by the observation equation $Y_t=o(X_t,Z_t)$,
where $Z_t$ is a random variable corresponding to an uncontrolled disturbance within the observation process and takes values in a finite set $\mathcal{Z}$. The sequence $\{Z_t: t \in \mathcal{T}\}$ consists of independent random variables that are also independent of $X_0$ and $\{W_t: t \in \mathcal{T}\}$. 

The system output $Y_t$ is received by both the human and the AI platform at each $t \in \mathcal{T}$. 
The platform generates a recommendation for the human with the goal of guiding the human's eventual action. Thus, this recommendation is a random variable $U_t^{\text{ai}}$ that takes values in the human's space of feasible actions $\mathcal{U}$.
At each $t \in \mathcal{T}$, the platform provides $U_t^{\text{ai}}$ based on the history $H_t = (Y_{0:t}, U_{0:t-1}^{\text{h}}, U_{0:t-1}^{\text{ai}}) \in \mathcal{H}_t$ 
and the recommendation strategy $\boldsymbol{g}^{\text{ai}} = (g_0^{\text{ai}}, \dots, g_T^{\text{ai}})$, where each recommendation law is the mapping $g_t^{\text{ai}}: \mathcal{H}_t \to \mathcal{U}$. Thus, the recommendation is $U_t^{\text{ai}} = g_t^{\text{ai}}(H_t)$ for all $t \in \mathcal{T}$.  

At each $t \in \mathcal{T}$, the human receives the recommendation before deciding which action to implement.
This decision is also affected by their own internal state, denoted by the random variable $S_t$ taking values in a finite space $\mathcal{S}$.
An internal state represents a combination of the human's interpretation of the system state, amenability towards AI suggestions, self-confidence, or a variety of other factors affecting the human's choices.
Starting at $S_0 \in \mathcal{S}$, the internal state evolves for all $t \in \mathcal{T}$ as $S_{t+1} = f^{\text{h}}(S_{t}, U_t^{\text{ai}}, Y_{t+1}, N_t)$, 
where $N_t$ is an uncontrolled disturbance that takes values in a finite set $\mathcal{N}$ and represents stochastic uncertainties in the evolution of the human's internal state. The initial internal state $S_0$ is independent of $X_0$ and the sequences $\{Z_t, W_t: t \in \mathcal{T}\}$.
Then, the human uses a control law $g^{\text{h}} : \mathcal{S} \times \mathcal{U} \to \mathcal{U}$ to implement the action
$U_t^{\text{h}} = g^{\text{h}}(S_t, U_t^{\text{ai}})$ at each $t \in \mathcal{T}$.
Subsequently, both the human and the AI platform receive shared feedback from the system, generated using the reward function $r: \mathcal{X} \times \mathcal{U} \to [r^{\min}, r^{\max}]$, where $r^{\min}, r^{\max} \in \mathbb{R}$. We denote this feedback by the random variable $R_t = r(X_t, U_t^{\text{h}}) = r\big(X_t,g^{\text{h}}(S_t, U_t^{\text{ai}})\big)$.
The objective of the AI platform is to maximize the expected total discounted reward:
    \begin{align} \label{eq:performance_criterion}
       J(\boldsymbol{g}^{\text{ai}}) =  \mathsf{E}^{\boldsymbol{g}^{\text{ai}}}\Bigg[\sum_{t=0}^{T} \gamma^{t} {\cdot} r\big(X_t,g^{\text{h}}(S_t, U_t^{\text{ai}})\big) \Bigg],
    \end{align}  
where $\mathsf{E}^{\boldsymbol{g}^{\text{ai}}}[\cdot]$ is the expectation with respect to the joint distribution imposed by strategy $\boldsymbol{g}^{\text{ai}}$, when human actions use the control law $g^{\text{h}}$, and $\gamma \in (0,1)$ is a discount factor. 

\begin{figure}[t!]
    \centering
    \includegraphics[width = \linewidth]{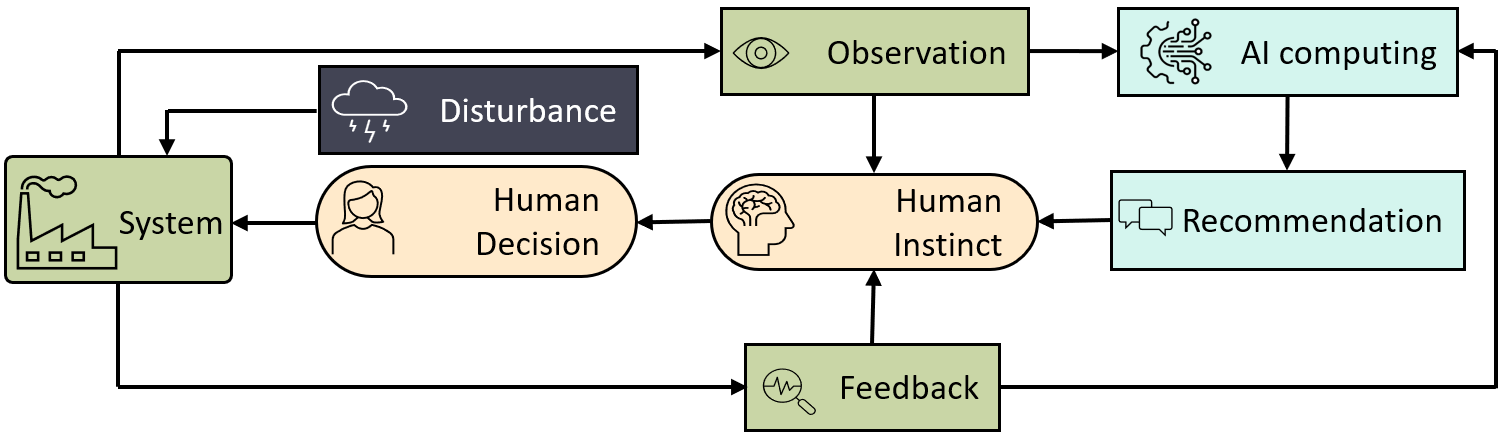}
    \caption{Control loop of the recommendation problem.}
    \label{fig:problem_setup}
    \vspace{-14pt}
\end{figure}

\begin{problem} \label{problem_1}
    The AI platform seeks an optimal recommendation strategy $\boldsymbol{g}^{*\text{ai}}$, such that $J(\boldsymbol{g}^{*\text{ai}}) \geq J(\boldsymbol{g}^{\text{ai}})$, given the sets $\{\mathcal{X},\mathcal{W},\mathcal{U}, \mathcal{Y}, \mathcal{Z}\}$ and functions $\{f, o\}$.
\end{problem}

An optimal strategy $\boldsymbol{g}^{*\text{ai}}$ exists because all variables are finite valued, but it may not be computable without knowledge of $\mathcal{S}$, $g^{\text{h}}$, and $f^{\text{h}}$. We impose the following assumptions.

\begin{assumption} \label{assumption_1}
    The human and the AI platform receive the same observation $Y_t$ at any $t \in \mathcal{T}$.
\end{assumption}

Assumption \ref{assumption_1} implies that the human cannot have more information than the AI platform at any $t$. 
In most CPHS applications, this assumption holds due to the AI platform's ability to access and assimilate large quantities of data.

\begin{assumption} \label{assumption_2}
    The action of the human $U_t^{\text{h}}$ and the reward $R_t$ are perfectly observed by the AI platform at each $t \in \mathcal{T}$.
\end{assumption}

Assumption \ref{assumption_2} implies that the human and the AI platform receive consistent rewards.
This assumption is required for the platform to anticipate human behavior. We anticipate the need for additional analysis in applications where humans may interpret rewards differently to a platform, e.g., economic systems \cite{annaswamy2023human}.

\vspace{-6pt}

\section{Recommendation Framework} \label{section:recommendation_framework}

In this section, we develop our theoretical framework to compute optimal recommendations. In Subsection \ref{subsection:known_internal_dynamics}, we analyze an AI platform with access to the true model for a human's behavior. This analysis yields a structural form for optimal AI recommendations. Building upon this structure and taking inspiration from recent work in partially observed reinforcement learning \cite{subramanian2022approximate, Dave2023infhorizon}, we define the notion of an approximate human model (AHM) in Subsection \ref{subsection:approximate_human}. We show that an AI platform can use an AHM to compute recommendations with performance guarantees. Finally, we propose an approach to construct an AHM in Subsection \ref{subsection:learning_algorithm}.

\vspace{-11pt}

\subsection{Optimal recommendation strategies} \label{subsection:known_internal_dynamics}

We start our exposition by considering that the AI platform knows a priori an exact human model consisting of the set of internal states $\mathcal{S}$, an initial distribution on $S_0$, the function $f^{\text{h}}(\cdot)$, and the human's control law $g^{\text{h}}(\cdot)$. However, the platform does not observe $S_t$ at any $t \in \mathcal{T}$. Next, we prove that such a system constitutes a partially observable Markov decision process (named the human-AI POMDP) for the platform.

\begin{lemma} \label{lemma:POMDP}
    Given a human model, Problem \ref{problem_1} is equivalent to computing the optimal strategy in a POMDP with state $(X_t, S_t) \in \mathcal{X} \times \mathcal{S}$, input $U_t^{\text{\emph{ai}}} \in \mathcal{U}$, observation $(Y_t, U_{t-1}^{\text{\emph{h}}}) \in \mathcal{Y} \times \mathcal{U}$, and reward $R_t \in [r^{\min}, r^{\max}]$ for all $t \in \mathcal{T}$.
\end{lemma}

\begin{proof}
    We establish that $\mathcal{X} \times \mathcal{S}$ is the state space for the POMDP by showing that (1) it predicts the reward and (2) the joint distribution on the next state and observation.
    For (1), recall from Section \ref{section:problem_formulation} that $R_t $ $= r(X_t, g^{\text{h}}(S_t, U_t^{\text{ai}}))$ at each $t \in \mathcal{T}$.
    For (2), for all $t \in \mathcal{T}$, consider any jointly feasible realization $(x_{0:t}, s_{0:t}, y_{0:t}, u_{0:t-1}^{\text{h}}, u_{0:t}^{\text{ai}})$ of the associated random variables. Using the law of total probability and Bayes' law, we state that the probability $\mathsf{P}(x_{t+1}, s_{t+1}, y_{t+1}, u^{\text{h}}_t~|~ x_{0:t}, s_{0:t}, y_{0:t}, u_{0:t-1}^{\text{h}}, u_{0:t}^{\text{ai}}) = \mathsf{I}\big(u_{t}^{\text{h}} = g^{\text{h}}(s_{t}, u_{t}^{\text{h}})\big) \cdot \mathsf{P}^{f^{\text{h}}}(s_{t+1}~|~s_t, y_{t+1}, u_{t}^{\text{ai}}) \cdot \mathsf{P}(y_{t+1}~|~x_{t+1}) \cdot \mathsf{P}^{g^{\text{h}}}(x_{t+1}~|~x_t, g^{\text{h}}(s_t, u_{t}^{\text{ai}}) = \mathsf{P}(x_{t+1}, s_{t+1}, y_{t+1}, u_{t}^{\text{h}}~|~x_t, s_t, u_t^{\text{ai}})$, where $\mathsf{I}(\cdot)$ is the indicator function. 
    Thus, $\mathcal{X} \times \mathcal{S}$ is a valid state space for the POMDP. Finally, the expected total discounted reward under any strategy $\boldsymbol{g}^{\text{ai}}$ in this POMDP is the same as \eqref{eq:performance_criterion}, implying that the human-AI POMDP yields the solution to Problem \ref{problem_1}.
\end{proof}




We can construct a dynamic programming (DP) decomposition for the human-AI POMDP in Lemma \ref{lemma:POMDP} using the history $H_t$ at each $t \in \mathcal{T}$. To this end, for all $h_t \in \mathcal{H}_t$ and $u_t^{\text{ai}} \in \mathcal{U}$, for all $t \in \mathcal{T}$, we recursively define the value functions
\begin{align}
    Q_{t}(h_t, u_t^{\text{ai}}) &:= \mathsf{E}\big[r(X_t, U_t^{\text{h}}) + \gamma \cdot V_{t+1}(H_{t+1})~|~h_t, u_t^{\text{ai}} \big], \\ 
    V_{t}(h_t) &:= \min_{u_t^{\text{ai}} \in \mathcal{U}} Q_{t}(h_t, u_t^{\text{ai}}),
\end{align}
where, $V_{T+1}(h_{T+1}) := 0$ identically, $U_t^{\text{h}} = g^{\text{h}}(S_t, u_t^{\text{ai}})$, and $H_{t+1} = (H_t, Y_{t+1},$ $ U_t^{\text{h}}, U_t^{\text{ai}})$ for all $t$. The recommendation law computed by this DP at each $t \in \mathcal{T}$ is $g^{*{\text{ai}}}_t(h_t) := \arg\min_{u_t^{\text{ai}}} Q_{t}(h_t, u_t^{\text{ai}})$. Standard arguments for POMDPs can be used to prove that the resulting recommendation strategy $\boldsymbol{g}^{*\text{ai}} := g^{*{\text{ai}}}_{0:T}$ is an optimal solution to the POMDP and consequently, to Problem \ref{problem_1} \cite{subramanian2019approximate}. However, this DP decomposition suffers from an increase in computational complexity as the history grows in size with time $t$. Furthermore, it does not provide insights into the underlying structure of optimal recommendation strategies. 
Typically, these challenges are overcome in POMDPs using an information state that compresses the history into a sufficient statistic \cite{Malikopoulos2021}. 
Thus, we construct an information state for the human-AI POMDP. To begin, we define two sufficient statistics for all $t \in \mathcal{T}$: (1) the AI's belief on the internal state $B^{\text{s}}_t := \mathsf{P}(S_t\,|\,H_t) \in \Delta(\mathcal{S})$, and (2) the AI's belief on the system state $B^{\text{x}}_t := \mathsf{P}(X_t\,|\,H_t) \in \Delta(\mathcal{X})$.
Note that the sufficient statistics are each a conditional probability distribution taking values in the space of distributions. 
We denote their realizations as $b_t^{\text{s}} \in \Delta(\mathcal{S})$ and $b_t^{\text{x}} \in \Delta(\mathcal{X})$, respectively and prove two important properties of the sufficient statistics. 

\begin{lemma} \label{lemma:independence}
    For any given realization $h_t \in \mathcal{H}_t$ of the history at time $t \in \mathcal{T}$, the internal state and system state are conditionally independent, i.e., for any $s_t \in \mathcal{S}$ and $x_t \in \mathcal{X}$:
    \begin{gather}
        \mathsf{P}(s_t, x_t\,|\,h_t) = \mathsf{P}(s_t\,|\,h_t){\cdot} \mathsf{P}(x_t\,|\,h_t) = b_t^{\text{\emph{s}}}(s_t){\cdot}b_t^{\text{\emph{x}}}(x_t).
    \end{gather}
\end{lemma}
\begin{proof}
    Let $h_t \in \mathcal{H}_t$, $s_t \in \mathcal{S}$ and $x_t \in \mathcal{X}$ denote the realizations of the associated random variables for all $t \in \mathcal{T}$. We prove the result using mathematical induction. The result holds trivially at $t=0$ since $S_0$ and $X_0$ are independent of each other. We assume that $\mathsf{P}(s_t, x_t\,|\,h_t) = b_t^{\text{x}}(x_t){\cdot}b_t^{\text{s}}(s_t)$ for some $t \in \mathcal{T}$. Then, at time $t+1$, we use Bayes' law to write
    \begin{align} \label{eq:proof_2_1}
        \mathsf{P}(s_{t+1}, x_{t+1}\,|\,h_{t+1})
        = \dfrac{\mathsf{P}(s_{t+1}, x_{t+1}, y_{t+1}, u_t^{\text{h}}\,|\, h_t, u_t^{\text{ai}})}{\mathsf{P}(y_{t+1}, u_t^{\text{h}}\,|\, h_t, u_t^{\text{ai}})}.
    \end{align}
    Expanding the numerator of \eqref{eq:proof_2_1}, we obtain that
    $\mathsf{P}(s_{t+1}, x_{t+1}, y_{t+1}, u_t^{\text{h}}\,|\, h_t, u_t^{\text{ai}}) = \sum_{\tilde{s}_t} \mathsf{P}(\tilde{s}_t, s_{t+1}, x_{t+1}, y_{t+1},$ $ u_t^{\text{h}}\,|\, h_t, u_t^{\text{ai}})
    = \sum_{\tilde{s}_t} \mathsf{P}(s_{t+1}\,|\,\tilde{s}_t, h_t, u_t^{\text{ai}}, y_{t+1}) \cdot \mathsf{P}(\tilde{s}_t|h_t, u_t^{\text{ai}}) \cdot \mathsf{I}( u_t^\text{h} 
    = g^{\text{h}}(\tilde{s}_t, u_t^{\text{ai}})) \cdot 
    \sum_{\tilde{x}_t} \mathsf{P}(y_{t+1}~|~x_{t+1}) \cdot \mathsf{P}(x_{t+1}~|~\tilde{x}_t, u_t^{\text{h}}) \cdot \mathsf{P}(\tilde{x}_{t}~|~h_t, u_t^{\text{ai}})
    = \mathsf{P}(s_{t+1}, u_t^{\text{h}}~|~h_t,u_t^{\text{ai}}, y_{t+1}) \cdot \mathsf{P}(x_{t+1}, y_{t+1}~|~h_t, u_t^{\text{h}}, u_t^{\text{ai}}),$ where $\mathsf{I}(\cdot)$ is the indicator function. 
    Similarly, using Bayes' law in the denominator of \eqref{eq:proof_2_1}, $\mathsf{P}(y_{t+1}, u_t^{\text{h}}\;|\; h_t, u_t^{\text{ai}}) = \mathsf{P}(u_t^{\text{h}}\;|\; h_t, u_t^{\text{ai}}) \cdot \mathsf{P}(y_{t+1}\;|\; h_t,  u_t^{\text{h}}, u_t^{\text{ai}})$. Substituting in \eqref{eq:proof_2_1}, we obtain $\mathsf{P}(s_{t+1}, x_{t+1}\;|\;h_{t+1}) = \dfrac{\mathsf{P}(s_{t+1}, u_t^{\text{h}}\;|\;h_t,u_t^{\text{ai}}, y_{t+1})}{\mathsf{P}(u_t^{\text{h}}\;|\; h_t, u_t^{\text{ai}})}
    \cdot \dfrac{\mathsf{P}(x_{t+1}, y_{t+1}\;|\;h_t, u_t^{\text{h}}, u_t^{\text{ai}})}{\mathsf{P}(y_{t+1}\;|\; h_t,  u_t^{\text{h}}, u_t^{\text{ai}})} = \mathsf{P}(s_{t+1}\;|\;h_{t+1})\cdot \mathsf{P}(x_{t+1}\;|\;h_{t+1}) = b_{t+1}^{\text{s}}(s_{t+1})\cdot b_{t+1}^{\text{x}}(x_{t+1}).$
    Thus, the result holds by mathematical induction.
\end{proof}


\begin{lemma} \label{lemma:independent_updates}
   We can construct a function $\psi^{\text{\emph{s}}}: \Delta(\mathcal{S}) \times \mathcal{U} \times \mathcal{Y} \to \Delta(\mathcal{S})$ independent of the choice of $\boldsymbol{g}^{\text{\emph{ai}}}$, such that
    \begin{gather}
        B_{t+1}^{\text{\emph{s}}} = \psi^{\text{\emph{s}}}(B_t^{\text{\emph{s}}}, U_t^{\text{\emph{ai}}},Y_{t+1}), \quad \forall t \in \mathcal{T} \label{eq:b_s_evol},
    \end{gather}
    and a function $\psi^{\text{\emph{x}}}: \Delta(\mathcal{X}) \times \mathcal{U} \times \mathcal{Y} \to \Delta(\mathcal{X})$ independent of both $\boldsymbol{g}^{\text{\emph{ai}}}$ and $g^{\text{h}}$, such that
    \begin{gather}
        B_{t+1}^{\text{\emph{x}}} = \psi^{\text{\emph{x}}}(B_t^{\text{\emph{x}}}, U_t^{\text{\emph{h}}},Y_{t+1}), \quad \forall t \in \mathcal{T}. \label{eq:b_x_evol}
    \end{gather}
\end{lemma}

\begin{proof}
    For all $t \in \mathcal{T}$ and any realizations $s_{t+1} \in \mathcal{S}_t$ and $h_{t+1} = (h_t, y_{t+1}, u_{t}^{\text{h}}, u_{t}^{\text{ai}}) \in \mathcal{H}_{t+1}$, using the law of total probability we obtain $b_{t+1}^{\text{s}}(s_{t+1}) = \mathsf{P}(s_{t+1}\,|\,h_t, y_{t+1}, u_{t}^{\text{h}}, u_{t}^{\text{ai}}) 
    = \sum_{\tilde{s}_t}\mathsf{P}(s_{t+1}\,|\,\tilde{s}_t, u_t^{\text{ai}}, y_{t+1}) \cdot \mathsf{P}(\tilde{s}_t\,|\,h_t)
    =: \psi^{\text{s}}(b_t^{\text{x}}, u_t^{\text{ai}}, y_{t+1})(s_{t+1}).$ Thus, we can construct $\psi^{\text{s}}$ that satisfies \eqref{eq:b_s_evol} independent of the choice of $\boldsymbol{g}^{\text{ai}}$.

    Similarly, for all $t \in \mathcal{T}$ and any realizations all $x_{t+1} \in \mathcal{X}$ and $h_{t+1} = (h_t, y_{t+1}, u_{t}^{\text{h}}, u_{t}^{\text{ai}}) \in \mathcal{H}_{t+1}$, using Bayes' law we obtain $b_{t+1}^{\text{x}}(x_{t+1}) = \dfrac{\mathsf{P}(x_{t+1}, y_{t+1}|h_t, u_t^{\text{h}}, u_t^{\text{ai}})}{\sum_{\bar{x}_{t+1}}\mathsf{P}(\bar{x}_{t+1}, y_{t+1}|h_t, u_t^{\text{h}}, u_t^{\text{ai}})}$. Both the numerator and denominator satisfy $\mathsf{P}(x_{t+1}, y_{t+1}|h_t, u_t^{\text{h}}, u_t^{\text{ai}}) =
    \sum_{\tilde{x}_t} \mathsf{P}(y_{t+1}|x_{t+1}) \cdot \mathsf{P}(x_{t+1}|\tilde{x}_t, u_t^{\text{h}}) \cdot b_t^{\text{x}}(\tilde{x}_t)$, hence, since they are only functions of $b_t^{\text{x}}$, $u_t^{\text{h}}$, and $y_{t+1}$, we can construct a function $\psi^{\text{x}}$ satisfying \eqref{eq:b_x_evol} independent of $\boldsymbol{g}^{\text{ai}}$ and $g^{\text{h}}$.   
\end{proof}

Next, we show that an information state for the human-AI POMDP is $\Pi_t := (B^{\text{s}}_t, B^{\text{x}}_t)$ for all $t \in \mathcal{T}$. We begin by establishing that $\Pi_t$ is sufficient to evaluate the expected cost.

\begin{lemma} \label{lemma:expected_reward}
    For all $t \in \mathcal{T}$, given realizations $h_t \in \mathcal{H}_t$, $u_t^{\text{\emph{ai}}} \in \mathcal{U}$, and $\pi_t = (b_t^{\text{\emph{s}}}, b_t^{\text{\emph{x}}})$, the expected conditional cost satisfies $\mathsf{E}[r(X_t, U_t^{\text{h}})\,|\,h_t, u_t^{\text{\emph{ai}}}] = \mathsf{E}[r(X_t, U_t^{\text{\emph{h}}})\,|\,\pi_t, u_t^{\text{\emph{ai}}}].$
\end{lemma}

\begin{proof}
    At any $t \in \mathcal{T}$, we state that
    $\mathsf{E}[r(X_t, U_t^{\text{h}})\,|\,h_t, u_t^{\text{ai}}] = \sum_{x_t, s_t} r(x_t,$ $ g^{\text{h}}(s_t, u_t^{\text{ai}})){\cdot} \mathsf{P}(x_t\,|\,h_t, u_t^{\text{ai}}) {\cdot}\mathsf{P}(s_t\,|\,h_t, u_t^{\text{ai}}) = \sum_{x_t, s_t} r(x_t, g^{\text{h}}(s_t, u_t^{\text{ai}})) {\cdot} b_t^{\text{x}}(x_t) {\cdot} b_t^{\text{s}}(s_t) = \mathsf{E}[r(X_t, U_t^{\text{h}})\,|\,\pi_t, u_t^{\text{ai}}]$, where, in the second equality, we use Lemma \ref{lemma:independence} and note that $S_t$ and $X_t$ are each independent of $U_t^{\text{ai}}$ given $H_t$.
\end{proof}

Next, we show that $\Pi_t$ is sufficient to predict the next observations in the human-AI POMDP at each $t \in \mathcal{T}$.

\begin{lemma} \label{lemma:next_pred}
    For all $t \in \mathcal{T}$, for any realizations $h_t \in \mathcal{H}_t$ and $u_t^{\text{\emph{ai}}} \in \mathcal{U}$, the corresponding realization $\pi_t$ of $\Pi_t$ satisfies
    \begin{gather} \label{eq:next_pred}
        \mathsf{P}(Y_{t+1}, U_t^{\text{\emph{h}}}~|~h_t, u_t^{\text{\emph{ai}}}) = \mathsf{P}(Y_{t+1}, U_t^{\text{\emph{h}}}~|~\pi_t, u_t^{\text{\emph{ai}}}).
    \end{gather}
\end{lemma}

\begin{proof}
    To prove the result, consider the $y_{t+1} \in \mathcal{Y}$ and $u_t^{\text{h}} \in \mathcal{U}$ for any $t \in \mathcal{T}$. Using the law of total probability and Bayes' law, we can expand the probability in \eqref{eq:next_pred} as $\mathsf{P}(y_{t+1}, u_t^{\text{h}}\,|\,h_t, u_t^{\text{ai}}) 
    = \sum_{\tilde{x}_{t+1}, \tilde{x}_t} \mathsf{P}(y_{t+1}|\tilde{x}_{t+1}){\cdot}\mathsf{P}(\tilde{x}_{t+1}\,|\,\tilde{x}_t, u_t^{\text{h}})$
    $\sum_{\tilde{s}_{t}} \mathsf{I}[u_t^{\text{h}} = g_t^{\text{h}}(\tilde{s}_t, u_t^{\text{ai}})] {\cdot} b_t^{\text{x}}(\tilde{x}_t) {\cdot} b_t^{\text{s}}(\tilde{s}_t) 
    = \mathsf{P}(y_{t+1}, u_t^{\text{h}}\,|\,\pi_t, u_t^{\text{ai}})$, where we use Lemma \ref{lemma:independence} in the second equality.
\end{proof}

Using the preceding results, we establish that $\Pi_t$ is an information state that it yields an optimal DP decomposition.

\begin{theorem} \label{theorem:optimal_recommendation}
    For all $t \in \mathcal{T}$, the random variable $\Pi_t = (B^{\text{s}}_t, B^{\text{x}}_t)$ is an information state of the human-AI POMDP. 
    Furthermore, for all $\pi_t \in \Delta(\mathcal{S}) \times \Delta(\mathcal{X})$ and $u_t^{\text{\emph{ai}}} \in \mathcal{U}$, let  $\bar{Q}_{t}(\pi_t, u_t^{\text{\emph{ai}}}) := \mathsf{E}[r(X_t, U_t^{\text{\emph{h}}}) + \gamma \cdot \bar{V}_{t+1}(\Pi_{t+1})\,|\,\pi_t, u_t^{\text{\emph{ai}}}]$ and $
    \bar{V}_{t}(\pi_t) := \min_{u_t^{\text{\emph{ai}}} \in \mathcal{U}} \bar{Q}_{t}(\pi_t, u_t^{\text{\emph{ai}}})$,
    where $\bar{V}_{T+1}(\pi_{T+1}) := 0$. Then, an optimal recommendation law in Problem \ref{problem_1} is $\bar{g}_t^{*\text{\emph{ai}}}(\pi_t) \hspace{-2pt} := \hspace{-2pt} \arg\min_{u_t^{\text{\emph{ai}}}} \bar{Q}_{t}(\pi_t, u_t^{\text{\emph{ai}}})$ for all $t$.
\end{theorem}

\begin{proof}
    Lemmas \ref{lemma:independent_updates} - \ref{lemma:next_pred} establish that $\Pi_t$ is sufficient to evaluate the expected cost, evolves in a state-like manner, and is sufficient to predict future observations for all $t \in \mathcal{T}$, hence it satisfies the standard conditions reported in \cite[Definition 3]{subramanian2022approximate} of an information state. As a direct consequence of the properties of information states \cite[Theorem 5]{subramanian2022approximate} and Lemma \ref{lemma:POMDP}, the recommendation strategy $\bar{\boldsymbol{g}}^{*{\text{ai}}} = \bar{g}^{\text{ai}}_{0:T}$ is an optimal solution to Problem \ref{problem_1}.
\end{proof}

Theorem \ref{theorem:optimal_recommendation} establishes that there is no loss of optimality when the AI platform holds beliefs $B_t^{\text{x}}$ and $B_t^{\text{s}}$ independent of each other  and utilizes them to compute optimal recommendations at each $t \in \mathcal{T}$. 
In practice, the AI platform can compute $B_t^{\text{x}}$ for all $t$ given the system dynamics in Problem \ref{problem_1}. However, in most applications, the platform will not have access to an exact model for human behavior to compute or update $B_t^{\text{s}}$. Thus, in the next subsection, we define the notion of an AHM that can either be designed heuristically or learned from data. We show that the AI can use an AHM in conjunction with $B_t^{\text{x}}$ to compute approximately optimal recommendations.

\begin{remark}
    In Problem \ref{problem_1}, if the system's state $X_t$ is perfectly observed, i.e., $Y_t = X_t$, by the AI platform we can use the same sequence of arguments as in Theorem \ref{theorem:optimal_recommendation} to prove that $(B_t^{\text{s}}, X_t)$ is an information state for Problem \ref{problem_1}.
\end{remark}

\vspace{-12pt}

\subsection{Approximate human model} \label{subsection:approximate_human}

In this subsection, we define the notion of an AHM that can be used by an AI platform instead of an exact human model. 

\begin{definition} \label{definition:ahm}
An \textit{approximate human model} consists of a Borel space $\hat{\mathcal{S}}$, an evolution equation $\hat{\sigma}_t: \mathcal{H}_t \to \hat{\mathcal{S}}$, and a probability mass function $\hat{\mu}: \hat{\mathcal{S}} \times \mathcal{U} \to \Delta(\mathcal{U})$, such that the approximate internal state $\hat{S}_t := \hat{\sigma}_t(H_t)$ satisfies for all $t \in \mathcal{T}$:

\textit{1) Evolution in a belief-like manner:} There exists a function $\hat{\psi}^{\text{s}}: \hat{\mathcal{S}} \times \mathcal{U} \times \mathcal{Y} \to \hat{\mathcal{S}}$ independent of the choice of recommendation strategy $\boldsymbol{g}^{\text{ai}}$, such that
\begin{gather} \label{eq:ahm_condition_1}
    \hat{S}_{t+1} = \hat{\psi}^{\text{s}}(\hat{S}_t, U_t^{\text{ai}}, Y_{t+1}).
\end{gather}

\textit{2) Approximate prediction of human actions:} 
For any realization $h_t \in \mathcal{H}_t$ and $u_t^{\text{ai}} \in \mathcal{U}$, the probability distribution induced by $\hat{\mu}$ is such that for some $\varepsilon > 0$:
\begin{gather} \label{eq:ahm_condition_2}
    \delta^{\text{TV}}\Big(\mathsf{P}^{g^{\text{h}}}(U_t^{\text{h}}~|~h_t, u_t^{\text{ai}}), \hat{\mu}(U_t^{\text{h}}~|~\hat{\sigma}_t(h_t), u_t^{\text{ai}})\Big) \leq \varepsilon,
\end{gather}
where $\delta^{\text{TV}}(\cdot, \cdot)$ is the total variation distance and $\mathsf{P}^{g^{\text{h}}}(\cdot)$ is the conditional probability distribution induced on $U_t^{\text{h}}$ by the human's choice of control law $g^{\text{h}}$.
\end{definition}

\begin{remark} \label{remark:total_variation}
    The total variation distance between any two probability mass functions $\mathsf{P}$ and $\mathsf{Q}$ on a finite set $\mathcal{A}$ is defined as $\delta^{\text{TV}}(\mathsf{P}, \mathsf{Q}) := \frac{1}{2}\sum_{a \in \mathcal{A}}|\mathsf{P}(a) - \mathsf{Q}(a)|$.
\end{remark}

\begin{remark}
    The AHM is directly inspired by the properties of the belief $B_t^{\text{s}}$ in Subsection \ref{subsection:known_internal_dynamics}. The first property imposes the structure in Lemma \ref{lemma:independent_updates} and the second property is essential to approximate the results of Lemmas \ref{lemma:expected_reward} - \ref{lemma:next_pred} later in Lemma \ref{lemma:approximate_properties}.
\end{remark}

\begin{remark}
    From Definition \ref{definition:ahm}, any empirically designed or learned model qualifies as an AHM if it satisfies the conditions \eqref{eq:ahm_condition_1} and \eqref{eq:ahm_condition_2}. Note that \eqref{eq:ahm_condition_1} is an intrinsic property of the AHM and \eqref{eq:ahm_condition_2} can be verified using an empirical distribution constructed from sampled observations of $U_t^{\text{h}}$ in the absence of the true underlying distribution $\mathsf{P}^{g^{\text{h}}}(U_t^{\text{h}}~|~h_t, u_t^{\text{ai}})$.
\end{remark}

Given an AHM, we define the random variable $\hat{\Pi}_t := (\hat{S}_t, B_t^{\text{x}})$ for all $t \in \mathcal{T}$.
Next, we prove that $\hat{\Pi}_t$ approximates the information state of the human-AI POMDP at each $t$, and it yields an approximately optimal recommendation strategy using the following DP decomposition. For all $t \in \mathcal{T}$, for all $\hat{\pi}_t \in \hat{\mathcal{S}} \times \Delta(\mathcal{X})$ and $u_t^{\text{ai}} \in \mathcal{U}$, we recursively define
    \begin{align}
        \hat{Q}_{t}(\hat{\pi}_t, u_t^{\text{ai}}) := \mathsf{E}[&r(X_t, U_t^{\text{h}}) + \gamma \hat{V}_{t+1}(\hat{\Pi}_{t+1})\,|\,\hat{\pi}_t, u_t^{\text{ai}}], \label{eq:approx_DP_1} \\
        \hat{V}_{t}(\hat{\pi}_t) := \min_{u_t^{\text{ai}} \in \mathcal{U}} &\hat{Q}_{t}(\hat{\pi}_t, u_t^{\text{ai}}), \label{eq:approx_DP_2}
    \end{align}
    where $\hat{V}_{T+1}(\hat{\pi}_{T+1}) := 0$ identically. Then, the corresponding recommendation law is $\hat{g}_t^{*\text{ai}}(\hat{\pi}_t) \hspace{-2pt} := \hspace{-2pt} \arg\min_{u_t^{\text{ai}}} \hat{Q}_{t}(\hat{\pi}_t, u_t^{\text{ai}})$ for all $t \in \mathcal{T}$. Next, we prove an essential property.


\begin{lemma} \label{lemma:approximate_properties}
    At any $t \in \mathcal{T}$, for any realizations $h_t \in \mathcal{H}_t$ and $u_t^{\text{ai}} \in \mathcal{U}$, the corresponding $\hat{\pi}_t \in \hat{\mathcal{S}} \times \Delta(\mathcal{X})$ satisfies:
    \begin{align}
        \hspace{-5pt} \text{\emph{a)} } &\big|\mathsf{E}^{g^{\text{\emph{h}}}}[r(X_t, U_t^{\text{\emph{h}}})\,|\,h_t, u_t^{\text{\emph{ai}}}] - \mathsf{E}^{\hat{\mu}}[r(X_t, U_t^{\text{\emph{h}}})\,|\,\hat{\pi}_t, u_t^{\text{\emph{ai}}}]\big| \nonumber \\ 
        & \hspace{150pt} \leq 2r^{\max} \cdot \varepsilon, \label{eq:approximate_reward} \\
        \hspace{-5pt} \text{\emph{b)} } &\delta^{\text{\emph{TV}}}\big(\mathsf{P}^{g^{\text{h}}}(Y_{t+1}, U_t^{\text{h}}|h_t, u_t^{\text{ai}}), \mathsf{P}^{\hat{\mu}}(Y_{t+1}, U_t^{\text{h}}|\hat{\pi}_t, u_t^{\text{ai}})\big) 
        \leq \varepsilon. \label{eq:approximate_observation}
    \end{align}
\end{lemma}

\begin{proof}
    At any $t \in \mathcal{T}$, for a given realization $h_t \in \mathcal{H}_t$ of the history, $\hat{\pi}_t = \big(\hat{\sigma}_t(h_t), b_t^{\text{x}}\big)$, where $b_t^{\text{x}} = \mathsf{P}(X_t|h_t)$. 
    
    a) To prove \eqref{eq:approximate_reward}, we expand the expected rewards under the distributions generated by $\mathsf{P}^{\text{g}^{\text{h}}}$ and $\hat{\mu}$, i.e.,
    $|\mathsf{E}^{g^{\text{h}}}[r(X_t, U_t^{\text{h}})~|~h_t, u_t^{\text{ai}}] - \mathsf{E}^{\hat{\mu}}[r(X_t, U_t^{\text{h}})~|~\hat{\sigma}_t(h_t), b_t^{\text{x}}, u_t^{\text{ai}}]| = |\sum_{\tilde{u}_t^{\text{h}}, \tilde{x}_t}r(\tilde{x}_t, \tilde{u}_t^{\text{h}}) \cdot b_t^{\text{x}}(\tilde{x}_t) \cdot \mathsf{P}^{g^{\text{h}}}(u_t^{\text{h}}~|~h_t, u_t^{\text{ai}})   - \sum_{\tilde{u}_t^{\text{h}}, \tilde{x}_t} r(\tilde{x}_t, \tilde{u}_t^{\text{h}}) \cdot b_t^{\text{x}}(\tilde{x}_t) \cdot \hat{\mu}(u_t^{\text{h}}~|~\hat{\sigma}(h_t), u_t^{\text{ai}}) | \leq 2 r^{\max} \cdot \varepsilon,$
    where, in the inequality, we use $b_t^{\text{x}}(\tilde{x}_t) = \mathsf{P}(\tilde{x}_t|h_t) \leq 1$ for all $t$, the definition of total variation distance in Remark \ref{remark:total_variation}, and the fact that $r^{\max}$ is an upper bound on the reward.

    b) To prove \eqref{eq:approximate_observation}, we first use the definition of the total variation distance and Bayes' law to write that
        $\delta^{\text{TV}}(\mathsf{P}^{g^{\text{h}}}(Y_{t+1}, U_t^{\text{h}}~|~h_t, u_t^{\text{ai}}), \mathsf{P}^{\hat{\mu}}(Y_{t+1}, U_t^{\text{h}}~|~\hat{\pi}_t, u_t^{\text{ai}})) \nonumber
        = \sum_{\tilde{y}_{t+1}, \tilde{u}_t^{\text{h}}} \frac{1}{2} |\mathsf{P}^{g^{\text{h}}}(\tilde{y}_{t+1}, \tilde{u}_t^{\text{h}}~|~h_t, u_t^{\text{ai}}) - \mathsf{P}^{\hat{\mu}}(\tilde{y}_{t+1}, \tilde{u}_t^{\text{h}}~|~\hat{\pi}_t, u_t^{\text{ai}})  | \nonumber
        = \sum_{\tilde{y}_{t+1}, \tilde{u}_t^{\text{h}}} \frac{1}{2} | \mathsf{P}^{g^{\text{h}}}(\tilde{y}_{t+1}~|~h_t, \tilde{u}_t^{\text{h}}) \cdot \mathsf{P}^{g^{\text{h}}}(\tilde{u}_t^{\text{h}}~|~h_t, u_t^{\text{ai}}) \nonumber 
        \quad - \mathsf{P}^{\hat{\mu}}(\tilde{y}_{t+1}~|~\hat{\pi}_t, \tilde{u}_t^{\text{h}}) \cdot \hat{\mu}(\tilde{u}_t^{\text{h}}~|~\hat{\sigma}_t(h_t), u_t^{\text{ai}})|$.
    Here, note that $\mathsf{P}^{g^{\text{h}}}(y_{t+1}~|~h_t, \tilde{u}_t^{\text{h}}) = \sum_{\tilde{x}_{t+1}, \tilde{x}_t} \mathsf{P}(\tilde{y}_{t+1}|\tilde{x}_{t+1}) \cdot \mathsf{P}(\tilde{x}_{t+1}|\tilde{x}_t, u_t^{\text{h}}) \cdot \mathsf{P}^{g^{\text{h}}}(\tilde{x}_t|h_t, u_t^{\text{h}}) = \sum_{\tilde{x_{t+1}}, \tilde{x}_t} \mathsf{P}(\tilde{y}_{t+1}|\tilde{x}_{t+1}) \cdot \mathsf{P}(\tilde{x}_{t+1}|\tilde{x}_t, u_t^{\text{h}}) \cdot b_t^{\text{x}}(\tilde{x}_t) = \mathsf{P}({y}_{t+1}|\hat{\pi}_t, u_t^{\text{h}}) = \mathsf{P}^{\hat{\mu}}({y}_{t+1}|\hat{\pi}_t, u_t^{\text{h}}),$ where, in the second equality we use Lemma \ref{lemma:independent_updates} to conclude that $b_t^{\text{x}}$ is independent of the choice of $g^{\text{h}}$; in the third equality, we use the fact that $\hat{\pi}_t$ contains $b_t^{\text{x}}$ as a component; and in the fourth equality, we use the same arguments to show that the probability is independent of the choice of $\hat{\mu}$. Substituting this result, we have that
    $\delta^{\text{TV}}\big(\mathsf{P}^{g^{\text{h}}}(Y_{t+1}, U_t^{\text{h}}~|~h_t, u_t^{\text{ai}}), \mathsf{P}^{\hat{\mu}}(Y_{t+1}, U_t^{\text{h}}~|~\hat{\pi}_t, u_t^{\text{ai}})\big) 
    \leq \hspace{-1pt} \frac{1}{2} \sum_{\tilde{y}_{t+1}, \tilde{u}_t^{\text{h}}} \hspace{-1pt} \mathsf{P}(\tilde{y}_{t+1}|\hat{\pi}_t, \tilde{u}_t^{\text{h}}) \cdot \big|\mathsf{P}^{g^{\text{h}}}(\tilde{u}_t^{\text{h}}~|~h_t, u_t^{\text{ai}}) -  \hat{\mu}(\tilde{u}_t^{\text{h}}~|~\hat{\sigma}_t(h_t),$ $ u_t^{\text{ai}})\big| 
    \leq \delta^{\text{TV}}\big(\mathsf{P}^{g^{\text{h}}}(U_t^{\text{h}}~|~h_t, u_t^{\text{ai}}), \hat{\mu}(U_t^{\text{h}}~|~\hat{\sigma}_t(h_t), u_t^{\text{ai}}) \big)$ $ \leq \varepsilon$, where in the second inequality we use Remark \ref{remark:total_variation} and note that $\mathsf{P}(\tilde{y}_{t+1}|\hat{\pi}_t, u_t^{\text{h}}) \hspace{-1pt} \leq \hspace{-1pt} 1$; and in the third inequality we use \eqref{eq:ahm_condition_2}. 
\end{proof}

Using Lemma \ref{lemma:approximate_properties}, we establish that the recommendation strategy $\hat{\boldsymbol{g}}_t^{*\text{\emph{ai}}} = \hat{g}_{0:t}^{*\text{\emph{ai}}}$ from \eqref{eq:approx_DP_1} - \eqref{eq:approx_DP_2} is approximately optimal.  

\begin{theorem} \label{theorem:approximation_garuantees}
    Let $||\hat{V}||_{\infty}$ be an upper bound on $\hat{V}_t(\hat{\pi}_t)$ for all $\hat{\pi}_t$ and $t \in \mathcal{T}$.
    Then, $\hat{\boldsymbol{g}}_t^{*\text{\emph{ai}}}$ is an approximately optimal recommendation strategy in Problem \ref{problem_1} with an optimality gap of at most $4 \varepsilon \cdot \big(r^{\max} + \sum_{t=1}^T \gamma^{t} \cdot (||\hat{V}||_{\infty} + r^{\max}) \big)$.
\end{theorem}

\begin{proof}
    Lemma \ref{lemma:approximate_properties} establishes that the random variable $\hat{\Pi}_t = (\hat{S}_t, B_t^{\text{x}})$ is sufficient to approximately evaluate the expected cost in \eqref{eq:approximate_reward} and is sufficient to approximately predict future observations in \eqref{eq:approximate_observation} for all $t \in \mathcal{T}$.  
    Furthermore, from \eqref{eq:ahm_condition_1} in Definition \ref{definition:ahm} and \eqref{eq:b_x_evol} in Lemma \ref{lemma:independent_updates}, we conclude that $\hat{\Pi}_t$ evolves in a state-like manner, hence it satisfies the conditions reported in \cite[Definition 2]{subramanian2019approximate} to qualify as an $(\epsilon, \delta)$-approximate information state for the human-AI POMDP, with $\epsilon = 2r^{\max}\cdot \varepsilon$ and $\delta = \varepsilon$. The result follows by substituting $\epsilon$ and $\delta$ into the performance bounds for approximate information states in \cite[Theorem 3]{subramanian2019approximate}.
\end{proof}


\vspace{-12pt}

\subsection{Constructing an approximate human model} \label{subsection:learning_algorithm}

We use supervised learning to learn the AHM in Definition \ref{definition:ahm}. 
We assume that we can access multiple trajectories $(Y_{t+1}, U_{t}^{\text{h}}, U_{t}^{\text{ai}} : t \in \mathcal{T})$ generated using an exploratory AI strategy. Then, we select two function approximators as follows: \textbf{(1) The encoder} is a recurrent neural network (e.g., LSTM or GRU) denoted by $\phi: \hat{\mathcal{S}} \times \mathcal{Y} \times \mathcal{U} \to \hat{\mathcal{S}}$ whose hidden state will be treated as $\hat{S}_t$ at each $t \in \mathcal{T}$. Thus, the inputs to $\phi$ are $(\hat{S}_{t-1}, Y_t, U_{t-1}^{\text{ai}})$ and its output is $\hat{S}_t$. \textbf{(2) The decoder} is a feed-forward neural network $\rho: \hat{\mathcal{S}} \times \mathcal{U} \to \Delta(\mathcal{U})$, whose inputs at each $t \in \mathcal{T}$ are $(\hat{S}_t, U_{t}^{\text{ai}})$ and whose output is the conditional distribution $\hat{\mu}$, represented conveniently as a vector in the probability simplex $\Delta(\mathcal{U})$. We also select a training loss $L = - \sum_{t=0}^B \log(\hat{\mu}_t(U_t^{\text{h}}))$, where $\hat{\mu}(U_t^{\text{h}})$ is the probability of the specific realization $U_t^{\text{h}}$ in the distribution $\hat{\mu}$. This loss function approximates the Kullback–Leibler divergence between the true distribution and $\hat{\mu}$, which forms an upper bound on the total variation distance in \eqref{eq:ahm_condition_2} by Pinker's inequality. Then, we have the following approaches to construct and train an AHM:

\textit{1) Combining empirical models with learning:} The main idea is to \textit{empirically select} an AHM space $\hat{\mathcal{S}}$ and evolution equation $\hat{\psi}^{\text{s}}$. The choice of $\hat{\mathcal{S}}$ is based on factors affecting human behavior within a specific application. For example, consider the partial adherence model \cite{grand2022best,faros2023adherence}, where $\hat{S}_t$ is the human's adherence level at each $t$, or the opinion aggregation model \cite{balakrishnan2022improving}, where $\hat{S}_t$ is the human's self-confidence at each $t$.
Similarly, the choice of $\hat{\psi}^{\text{s}}$ is to ensure that $\hat{S}_{t+1} = \hat{\psi}^{\text{s}}(\hat{S}_{t}, U_t^{\text{ai}}, Y_{t+1})$ for all $t \in \mathcal{T}$. To learn our model, we feed $\hat{S}_t$ from the empirical model and $U_t^{\text{ai}}$ to the decoder $\rho$ at each $t$ and train $\rho$ over the trajectories with loss $L$.

\textit{2) Using only supervised learning:} When we cannot use domain knowledge, we learn an AHM from data by assuming an encoder-decoder architecture. We consider the encoder $\phi$ and feed its internal state $\hat{S}_t$ with $U_t^{\text{ai}}$ to the decoder $\rho$ at each $t \in \mathcal{T}$. We train the complete network assembly with loss $L$.

\vspace{-6pt}
\section{Numerical Example} \label{sec:example}

\begin{figure}[t!]
    \centering
    \includegraphics[width=\linewidth]{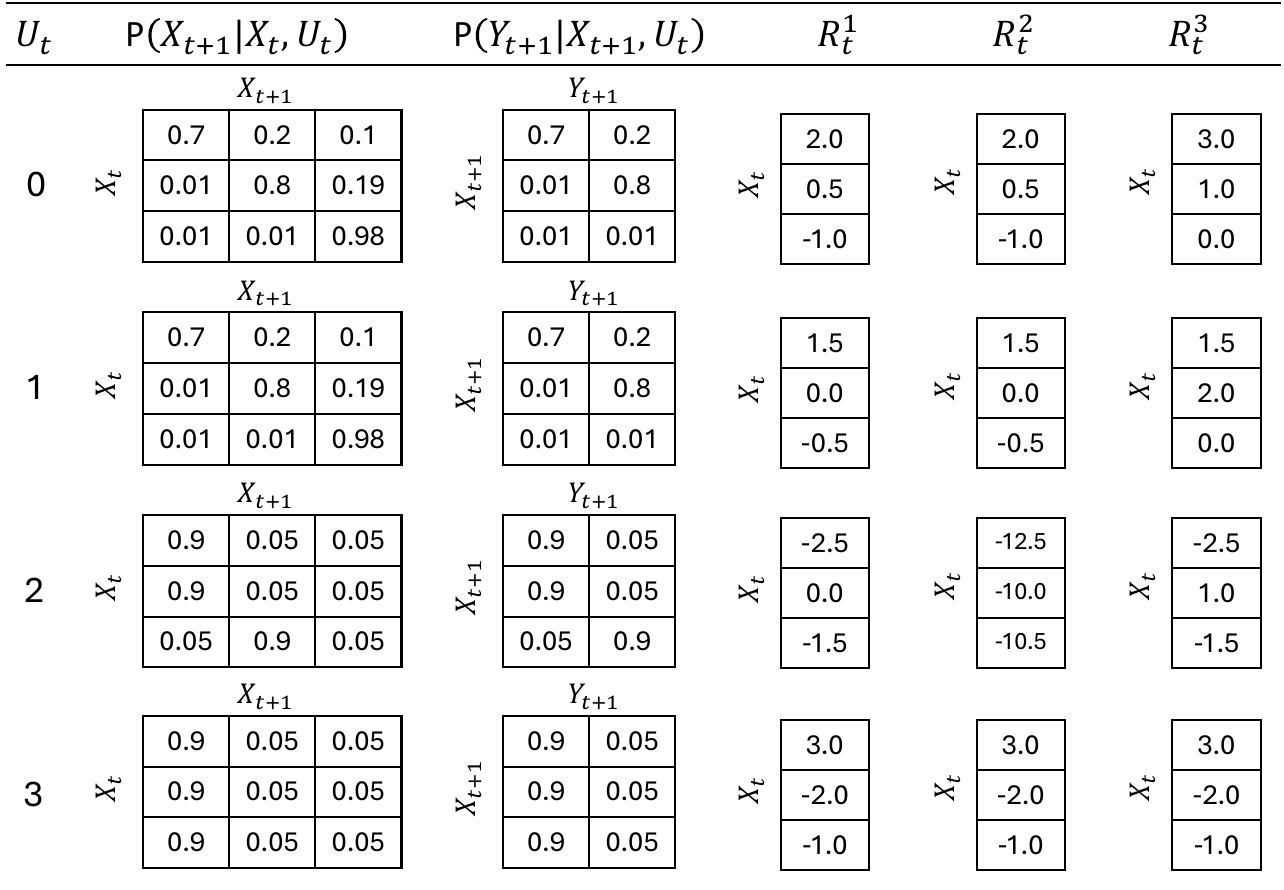}
    \caption{System model for the machine.}
    \label{fig:machine_model}
    \vspace{-14pt}
\end{figure}

In this section, we illustrate our results with a simple example.
We consider a partially observed machine replacement problem with a human operator who receives suggestions from an AI platform. The machine's state $X_t = \{0,1,2\}$ represents the number of failures at each $t \in \mathcal{T}$. The possible actions are $\mathcal{U} = \{0,1,2,3\}$, where $0$ is produce, $1$ is inspect, $2$ is small repair, and $3$ is major repair. At each $t \in \mathcal{T}$, the machine's state evolves using the transition probabilities in Fig. \ref{fig:machine_model}. The human-AI team receive an observation $Y_t \in \{0,1\}$ representing the quality of the machine output at each $t$ using the probabilities in Fig. \ref{fig:machine_model}. 
We consider a lazy human operator, whose internal state $S_t \in \{0,1\}$ denotes their motivation at any $t \in \mathcal{T}$. If $S_t = 1$ and $U_t^{\text{ai}} \in \{0,1,3\}$, the operator selects $U_t^{\text{h}} = U_t^{\text{ai}}$ with probability $0.97$ and selects any other action probability of $0.01$ each. However, if $S_t = 1$ and $U_t^{\text{ai}} = 2$, the lazy operator does not carry out minor repairs and instead decides to produce, i.e., $U_t^{\text{h}} = 0$.
Furthermore, if $U_t^{\text{ai}} = 3$, the operator does carry out major repairs but loses motivation, i.e., $U_t^{\text{h}} = 3$ and $S_{t+1} = 0$. In contrast, when $S_t = 0$, the operator almost always produces, i.e., $U_t^{\text{h}} = 0$ with probability $0.99$ and follows $U_t^{\text{h}} = U_t^{\text{ai}}$ with probability $0.01$. Furthermore, the operator recovers motivation after one time step, i.e., if $S_t = 0$ then $S_{t+1} = 1$. 
To incorporate interactions with the operator, we consider $3$ reward functions in Fig. \ref{fig:machine_model}. A natural reward is $R_t^1$, whereas $R_t^2$ discourages recommendation of $U_t^{\text{ai}} = 2$ and $R_t^3$ discourages $U_t^{\text{ai}} \in \{2,3\}$ and encourages $U_t^{\text{ai}} \in \{0,1\}$.

\begin{figure}[t!]
    \centering
    \includegraphics[width=\linewidth]{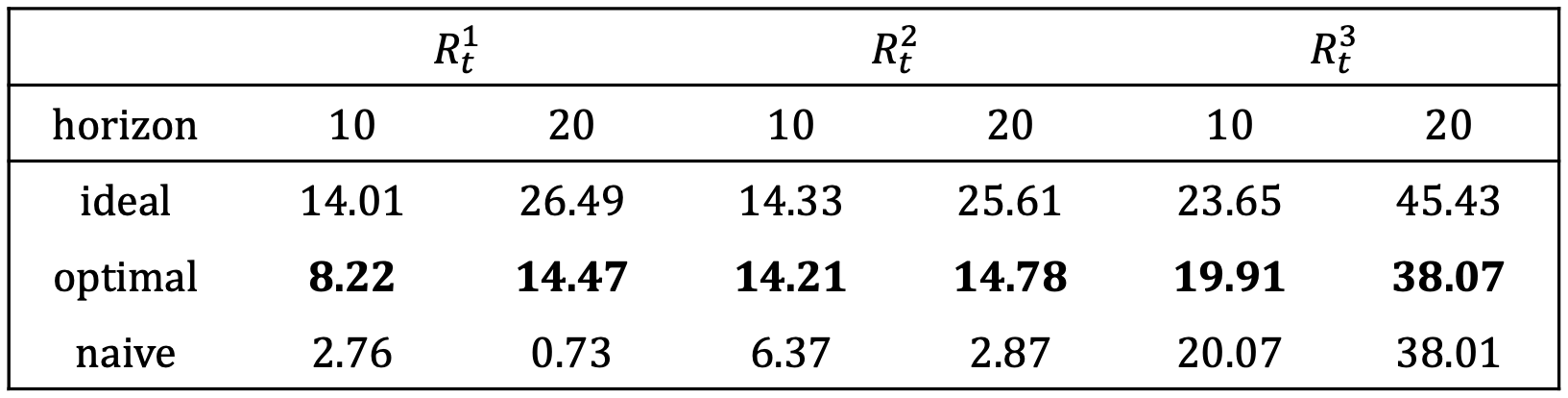}
    \caption{Rewards obtained using different strategies.}
    \label{fig:results}
    \vspace{-14pt}
\end{figure}

We construct an AHM using the first approach in Subsection \ref{subsection:learning_algorithm} and assuming $\hat{S}_t = (Y_t, A_{t-1}, A_{t-2})$, where $A_{t} = \mathsf{I}(U_{t}^{\text{h}} = U_t^{\text{ai}}) \in \{0,1\}$ indicates the adherence of the human to AI recommendations and $Y_t \in \{0,1\}$ is one-hot encoded. Note that $\hat{S}_t$ naturally satisfies \eqref{eq:ahm_condition_1}. The decoder $\rho$ has $4$ linear layers of sizes $(4,6)(6,8)(8,6)(6,4)$, where the first three layers have ReLU activation and the final layer has Sigmoid activation. We train decoder $\rho$ over $10,000$ trajectories with $T=50$ and a learning rate $0.0001$. Then, with discount $\gamma = 0.95$, we use the trained model to create the human-AI POMDP and compute the optimal recommendation strategy $\boldsymbol{g}^{*\text{ai}}$ using SARSOP \cite{kurniawati2008sarsop}. As a baseline, we also compute a naive AI strategy $\boldsymbol{g}^{*{\text{naive}}}$ \textit{without} considering a human in the loop using SARSOP. Our results are obtained by running $100$ simulations for time horizons $T=10$ and $T=20$ in three situations: \textbf{(1) ideal:} when $\boldsymbol{g}^{*{\text{naive}}}$ is implemented in a system without a human in the loop; \textbf{(2) optimal:} when $\boldsymbol{g}^{*{\text{ai}}}$ is implemented with a human; and \textbf{(3) naive:} when $\boldsymbol{g}^{*{\text{naive}}}$ is implemented with a human. We plot the actual rewards in Fig. \ref{fig:results}. The ideal case outperforms the others, indicating that the presence of a human may degrade performance. However, for both $R_t^1$ and $R_t^2$, the optimal case outperforms the naive case significantly, highlighting the utility of the learned AHM. In $R_t^3$, our rewards discourage $U_t^{\text{ai}} \in \{2,3\}$, and thus, both ideal and naive cases perform almost equally. Thus, the naive strategy and optimal strategy perform almost equally well. In $R_t^3$, for $T=10$, the errors within the learned AHM can explain the slight overperformance of the naive strategy over optimal.

\vspace{-6pt}

\section{Concluding Remarks} \label{sec:conclusion}

In this letter, we developed a framework for CPHS with partially observed data. We established the structural form of optimal recommendations and provided an AHM that can facilitate approximately optimal recommendations. Finally, we presented an approach to constructing AHMs from data and illustrated its utility in a numerical example. Future work should consider applying this framework to specific CPHS applications.

\vspace{-6pt}

\bibliographystyle{IEEEtran}
\bibliography{References, IDS}

\end{document}